\documentclass[runningheads]{llncs}

\usepackage[T1]{fontenc}
\usepackage{graphicx}
\usepackage{cite}
\usepackage{algorithm2e}
\usepackage{pxfonts}

\def\calP{\mathcal{P}}

\def\Ed{\mathsf{E}\mathrm{d}}
\usepackage[in]{fullpage}

\begin{document}

\title{The Two-Center Problem of Uncertain Points on Trees\thanks{A preliminary version of this paper appeared in \textit{Proceedings of the 16th Annual International Conference on Combinatorial Optimization and Applications (COCOA 2023)}}}

\titlerunning{The Two-Center Problem of Uncertain Points on Trees}
\author{Haitao Xu \and Jingru Zhang}

\authorrunning{H. Xu and J. Zhang}

\institute{Cleveland State University, Cleveland, Ohio 44115, USA \\
\email{h.xu12@vikes.csuohio.edu, j.zhang40@csuohio.edu}}
\maketitle             
%150--250 words
\begin{abstract}
 In this paper, we consider the (weighted) two-center problem of uncertain points on a tree. Given are a tree $T$ and a set $\calP$ of $n$ (weighted) uncertain points each of which has $m$ possible locations on $T$ associated with probabilities. The goal is to compute two points on $T$, i.e., two centers with respect to $\calP$, so that the maximum (weighted) expected distance of $n$ uncertain points to their own expected closest center is minimized. This problem can be solved in $O(|T|+ n^{2}\log n\log mn + mn\log^2 mn \log n)$ time by the algorithm for the general $k$-center problem. In this paper, we give a more efficient and simple algorithm that solves this problem in $O(|T| + mn\log mn)$ time. 
\end{abstract}

\section{Introduction}\label{sec:intro}
Facility locations play a significant role in operations research due to its wide applications in transportation, sensor deployments, circuit design, etc. Consider the inherent uncertainty of collected data caused by measurement errors, sampling discrepancy, and object mobility. It is natural to consider facility location problems on uncertain points. The locational model for uncertain points has been considered a lot in facility locations~\cite{ref:WangCo17,ref:QuanLi21,ref:HuCo23}, where the location of an uncertain point is represented by a probability density function (pdf). In this paper, we study the two-center problem, one of classical facility location problems, for uncertain points on a tree under the locational model.

% Two models are often employed to describe uncertain points: the locational model~\cite{ref:YiuEf09,ref:KamousiCl11} and the existence model~\cite{ref:ChengEf04,ref:TaoRa07}. The locational model is considered more often in facility locations an uncertain point exists but its location is uncertain and represented by a probability density function (pdf).   
% Note that the existence model is a special case of the locational model.
% In this paper, we study the two-center problem, one of classical optimization problems in facility locations, for uncertain points on a tree under the locational model.

Let $T$ be a tree. We consider each edge $e$ of $T$ as a line segment of a positive length so that we can talk about ``points'' on $e$. Formally, we specify a point $x$ of $T$ by an edge $e$ that contains $x$ and the distance between $x$ and an incident vertex of $e$. For any two points $p$ and $q$ on $T$, the distance $d(p,q)$ is the sum of all edges on the simple path from $p$ to $q$. Let $\calP$ be a set of $n$ uncertain points $P_1, \cdots, P_n$. Each $P_i\in\calP$ is associated with $m$ locations $p_{i1}, p_{i2},\cdots, p_{im}$ each being a point on $T$, and each location $p_{ij}$ has a probability $f_{ij}\geq 0$ for $P_i$ appearing at location $p_{ij}$. Additionally, each $P_i$ has a non-negative weight $w_i$. 

% See Fig.~\ref{fig:fig1}. 

% \begin{figure}[h]
% \centering
% \begin{minipage}{0.8\linewidth}
%  \centering
%  \includegraphics[width=0.9\textwidth]{Fig.1.eps}
%  \caption{Illustrating three uncertain points $P_1,P_2,P_3$, each of which has some possible locations and probabilities respectively. }
%  \label{fig:fig1}
% \end{minipage}
% \end{figure}

For any (deterministic) point $x$ on $T$, the distance of any uncertain point $P_i$ to $x$ is the {\em expected} version and defined as $\sum_{j=1}^{m}f_{ij}\cdot d(p_{ij}, x)$. 
% where $d(p_{ij}, x)$ is the length of the path between points $p_{ij}$ and $x$. 
We use $\Ed(P_i,x)$ to denote the expected distance of $P_i$ to $x$. Let $x_1$ and $x_2$ be any points on $T$. We say that an uncertain point $P_i$ is {\em expectedly closer} to $x_1$ if $\Ed(P_i,x_1)\leq\Ed(P_i,x_2)$. 

Define $\phi(x_1, x_2) = \max_{1\leq i\leq n}\{w_i\cdot\min(\Ed(P_i,x_1), \Ed(P_i,x_2))\}$. The {\em two-center} problem aims to compute two points on $T$ so as to minimize $\phi(x_1, x_2)$ and the two optimal points, denoted by $q^*_1$ and $q^*_2$, are called centers with respect to $\calP$ on $T$. We say that $P_i$ is {\em covered} by $q^*_1$ if $P_i$ is expectedly closer to $q^*_1$. 
% Denote by $q^*_1$ and $q^*_2$ denote the two centers. 

The algorithm~\cite{ref:WangCo19} for the general $k$-center problem can address our problem in $O(|T|+ n^{2}\log n\log mn + mn\log^2 mn\log n)$ time. In this paper, however, we present an $O(|T|+ mn\log mn)$-time algorithm with the assistance of our proposed linear-time approach for the decision two-center problem. Note that the time complexity of our algorithm almost matches the $O(|T|+ mn\log m +n\log n)$ result~\cite{ref:XuTh23} for the case of $T$ being a path.

% In details, Wang and Zhang~\cite{ref:WangCo19} first addressed in $O(|T| + mn\log^2 mn)$ time a coverage problem of computing minimum points on $T$ so that the expected distance of each $P_i$ to at least one point is at most a given value $\lambda\geq 0$. This algorithm can be employed to solve the {\em decision} $k$-center problem on $T$ with respect to $\calP$ that determines whether $k$ points can be found on $T$ so that the expected distance of each $P_i$ to its expectedly closest point is no more than the given $\lambda$. Further, they compute the optimal objective value of the $k$-center problem by performing a binary search among all $O(n^2)$ candidate values $w_i\Ed(P_i,x)= w_j\Ed(P_j,x)$. 

% In this paper, we present an $O(mn)$-time algorithm for the decision two-center problem. With the assistance of this algorithm, an $O(|T|+mn\log mn)$-time approach is proposed to compute $q^*_1$ and $q^*_2$ on $T$. Note that the time complexity of our algorithm almost matches the $O(|T|+ mn\log m +n\log n)$ result~\cite{ref:XuTh23} for the case of $T$ being a path. 

\subsection{Related work}\label{sec:prework}
If every $P_i$ has only one location then the problem falls into the deterministic case. Ben-Moshe et al.~\cite{ref:Ben-MosheAn06} adapted Megiddo's prune-and-search technique~\cite{ref:MegiddoLi83} to 
% for the deterministic one-center to 
solve in $O(n)$ time the deterministic two-center on a tree where each vertex is a demand point. On a cactus graph, the two-center problem was addressed in their another work~\cite{ref:Ben-MosheEf07}, and an $O(n\log^3 n)$-time algorithm was proposed. 
On a general graph, Bhattacharya and Shi~\cite{ref:BhattacharyaIm14} reduced the decision problem into the two-dimensional Klee's measure problem~\cite{ref:ChanKl13} so that the problem can be solved in polynomial time. %$O(|E|^2|V|\log^2 |V|)$ time where  . 
The planar version was studied in several works~\cite{ref:WangOn22,ref:ChanDy99,ref:EppsteinFa97}. The state-of-the-art result is an $O(n\log^2 n)$ deterministic algorithm given by Wang~\cite{ref:WangOn22}.
% which employs Cole's parametric search~\cite{ref:ColeSl87} to solve the problem with the assistance of their faster decision algorithm. 

In general, every $P_i$ has more than $m>1$ locations on $T$. As mentioned above, Wang and Zhang~\cite{ref:WangCo19} considered the general $k$-center problem so that the two-center problem can be addressed in $O(|T|+ n^{2}\log n\log mn + mn\log^2 mn\log n)$ time. If $T$ is a path, the two-center was solved in $O(|T|+ mn\log m +n\log n)$ time in our previous work~\cite{ref:XuTh23}. 
% which improves the result~\cite{ref:WangOn15} for the general one-dimensional $k$-center problem on uncertain points each with a piece-wise linear pdf. 
One of the most related problems is the one-center problem. Wang and Zhang~\cite{ref:WangCo17} generalized Megiddo's prune-and-search technique to solve the one-center of $\calP$ on a tree in linear time. Hu and Zhang~\cite{ref:HuCo23} studied the uncertain one-center problem on a cactus graph and proposed an $O(|T|+mn\log mn)$-time algorithm. Moreover, Li and Huang~\cite{ref:HuangSt17} considered the planar Euclidean uncertain $k$-center and gave an approximation algorithm. Later, Alipour and Jafari~\cite{ref:AlipourIm21} improved their result to an $O(3+\epsilon)$-approximation and proposed a $10$-approximation algorithm for any metric space. 

% Facility location problems on uncertain points under other models have attracted attentions as well. Foul~\cite{ref:FoulA106} studied the planar one-center problem for uncertain points each of which has a uniform distribution in a given rectangle. If each uncertain point could appear at any point of a region without specified pdf, L{\"{o}}ffler and van Kreveld~\cite{ref:LofflerLa10} solved the one-center problem, and very recently, Keikha et al.~\cite{ref:KeikhaCl21} studied the general $k$-center version. See also the minmax regret problems, e.g.,~\cite{ref:AverbakhMi97,ref:AverbakhFa05}.   

\subsection{Our approach}\label{sec:ourapp}
The locations of the uncertain points of $\calP$ may be in the interior of edges of $T$. A vertex-constrained case happens when all locations are at vertices of $T$ and each vertex of $T$ contains locations. As shown in~\cite{ref:WangCo19}, any general case can be reduced to a vertex-constrained case. In the following, we focus on discussing the vertex-constrained case. 

Let $\lambda^*$ be the minimized objective value. The median of each $P_i\in\calP$ is the point where $\Ed(P_i,x)$ reaches its minimum. We first solve the decision problem that determines whether $\lambda\geq\lambda^*$ for any given $\lambda$, i.e., decide if $\lambda$ is \textit{feasible}. This can be addressed in $O(mn\log^2 mn)$ time by the algorithm~\cite{ref:WangCo19}, which relies on several dynamic data structures requiring $O(mn\log^2 mn)$-time constructions. Since the convexity of each $\Ed(P_i,x)$ on any path, we develop a much simpler algorithm that is free of any data structures and runs faster in $O(mn)$ time.

% To improve the bound, the key is to prune irrelevant parabolas on functions to accelerate the search. 
% As discussed in Section~\ref{sec:decision}, the approach~\cite{ref:WangCo19} 
% replies on several dynamic data structures with $O(mn\log^2 mn)$-time constructions. T

% deciding if a center must be placed on an edge. % which places necessary centers in the bottom-up manner, % the convexity of $\Ed(P_i,x)$ on any path leads 

Regarding our decision problem, there is an observation that there exists an edge $e$ on $T$ such that a center must be placed on $e$ to cover all uncertain points whose medians are on one of the two subtrees generated by removing $e$ from $T$. Such an edge is called a \textit{peripheral-center} edge. Our decision algorithm first computes a peripheral-center edge in $O(mn)$ time. With this edge, the feasibility of $\lambda$ can be known in $O(mn)$ time. 
% Our algorithm is free of any data structures and thus much simpler. 

To compute centers $q^*_1$ and $q^*_2$, we first compute the two \textit{critical} edges that respectively contain $q^*_1$ and $q^*_2$. Because an essential lemma can decide in $O(mn)$ time which split subtree of any given point contains a critical edge. Our algorithm finds each critical edge on $T$ recursively with the assistance of this lemma. Once the two edges are found, $q^*_1$ and $q^*_2$ can be computed in $O(mn\log n)$ time. 

\section{Preliminaries}\label{sec:pre}
% In general, locations of every uncertain point could be anywhere on the given tree. If every location is at a vertex on $T$ and every vertex holds a location then the problem is indeed a \textit{vertex-constrained} case. Note that centers could be any points on $T$ in a vertex-constrained case. As shown in~\cite{ref:WangCo19}, any given general instance of the $k$-center problem can be reduced into a vertex-constrained instance in $O(|T|+mn)$ time. Our algorithm thus focus on solving the vertex-constrained version. 
% So, $|T|\leq mn$. 

Let $u$ and $v$ be any two vertices on $T$. Denote by $e(u,v)$ the edge incident to both $u$ and $v$. For any two points $p$ and $q$ on $T$, we let $\pi(p,q)$ be the simple path between $p$ and $q$. As in~\cite{ref:WangCo19}, the lowest common ancestor data structure~\cite{ref:BenderTh00} can be applied to $T$ so that with an $O(mn)$ preprocessing work, the path length of $\pi(p,q)$, i.e., the distance $d(p,q)$, can be known in constant time. 

Let $\pi$ be any simple path on $T$ and $x$ be any point on $\pi$. For any $P_i\in\calP$, as analyzed in~\cite{ref:WangCo17}, $\Ed(P_i,x)$ is a convex piece-wise linear function in $x\in\pi$, and it monotonically increases or decreases as $x$ moves on any edge of $\pi$ from one ending vertex to the other. 

\begin{figure}[ht]
\centering
\begin{minipage}{0.48\linewidth}
 \centering
 \includegraphics[width=0.65\textwidth]{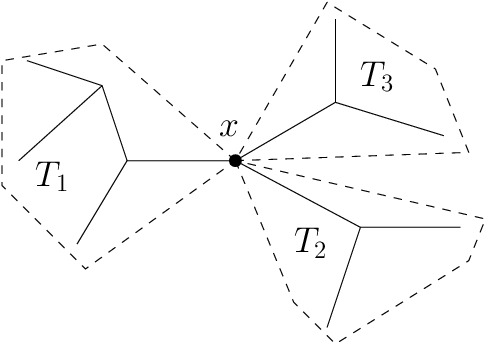}
 \caption{The point $x$ has three split subtrees $T_1$, $T_2$ and $T_3$.}
 \label{fig:fig1}
\end{minipage}
\hfill
\begin{minipage}{0.48\linewidth}
 \centering
 \includegraphics[width=0.95\textwidth]{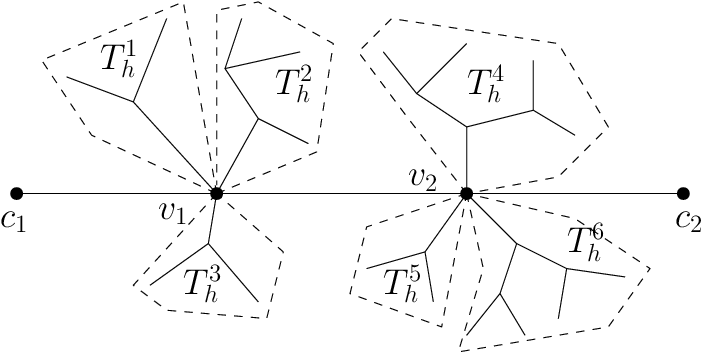}
 \caption{Illustrating the tree $T_h$ for the case $C=2$ where $V =\{v_1, v_2\}$ and $\Gamma(V)=\{T^1_h, \cdots, T^6_h\}$.}
\label{fig:fig2}
\end{minipage}
\end{figure}

Let $x$ be any point on $T$. Removing $x$ from $T$ generates several disjoint subtrees. Consider $x$ as an ``open vertex'' on each subtree 
that is free of any locations. We call each obtained subtree a \textit{split subtree} of $x$. See Fig.~\ref{fig:fig1} for an example. For any $P_i\in\calP$ and any subtree $T'$ of $T$, we refer to the sum of probabilities of $P_i$'s all locations in $T'$ as the \textit{probability sum} of $P_i$ in $T'$. Denote by $p_i^*$ the median of $P_i$ that minimizes $\Ed(P_i,x)$ among all points on $T$. 
% The median of $P_i$ may not be unique and let $p^*_i$ be any of them. 
The following lemma was given in~\cite{ref:WangCo17}. 

\begin{lemma}\label{lem:medians}
\cite{ref:WangCo17} Consider any point $x$ on $T$ and any uncertain point $P_i$ of $\calP$. 
\begin{enumerate}
\item If $x$ has a split subtree whose probability sum of $P_i$ is greater than $0.5$, then $p^*_i$ must be in that split subtree.
\item The point $x$ is $p^*_i$ if the probability sum of $P_i$ in each of $x$'s split subtrees is less than $0.5$. 
\item The point $x$ is $p^*_i$ if $x$ has a split subtree in which the probability sum of $P_i$ is equal to $0.5$. 
\end{enumerate}
\end{lemma}

Note that the median of $P_i$ may not be unique. But all points on $T$ minimizing $\Ed(P_i,x)$ induce a connected subtree of $T$. Let $p^*_i$ be any of these points. Due to the convexity of $\Ed(P_i,x)$, we have that $\Ed(P_i,x)$ monotonically increases as $x$ moves along any path away from $p^*_i$. 

\section{The decision algorithm}\label{sec:decision}
% In this section, we present our algorithm for solving the decision version of the problem. 
Given any value $\lambda>0$, the decision problem is to determine whether there exist no more than two points, i.e., centers, on $T$ so that the expected distance of each $P_i\in\calP$ to at least one of them is at most $\lambda$. If yes, then $\lambda\geq\lambda^*$ and so $\lambda$ is feasible. Otherwise, $\lambda<\lambda^*$ and $\lambda$ is infeasible. Clearly, $\lambda^*$ is the smallest feasible value. 

To establish a context for our work, we first introduce the algorithm~\cite{ref:WangCo19} for the decision $k$-center problem. Let $T_m$ be the minimum subtree of $T$ spanning the medians of $n$ uncertain points. We say that an uncertain point is covered by a center if their expected distance is at most $\lambda$. The convexity of $\Ed(P_i,x)$ leads the greedy algorithm that places minimum centers in the bottom-up manner whenever it ``has to''. During the post-order traversal, for each vertex $v$, considering all `active' uncertain points whose medians are in the subtree rooted at $v$, a center must be placed on the edge incident to $v$ and its parent vertex $u$ if one of these uncertain points has its expected distance at $u$ larger than $\lambda$. If so, then the center is placed on this edge at the point where the maximum expected distance of these uncertain points equals $\lambda$, and next all uncertain points covered by this center are ``deactivated''. With the assistance of several data structures, minimum centers can be placed on $T$ in $O(mn\log^2 mn)$ time. 
%as close to the root as possible to cover more uncertain points. and then 

We say that an edge $e$ of $T$ is a \textit{peripheral-center} edge if a center must be placed on $e$ so that this center is determined by all uncertain points whose medians are all on one of the two subtrees generated by removing $e$ from $T$ (but keeping its two incident vertices), and the center on $e$ is a \textit{peripheral} center. 

Regarding our problem, the goal is to determine whether two peripheral centers can be found on $T$ so that the expected distance of every uncertain point to at least one of them is no more than $\lambda$. Our algorithm thus first finds a peripheral-center edge and then applies
Lemma~\ref{lem:basecase} to decide the feasibility of the given $\lambda$. If it is feasible then the two centers are returned. 

\begin{lemma}\label{lem:basecase}
Given a peripheral-center edge $e$ on $T$, we can determine in $O(mn)$ time whether $\lambda\geq\lambda^*$, and if yes, the two centers can be computed in $O(mn)$ time. 
\end{lemma}
\begin{proof}
Let $u$ and $v$ be the two incident vertices of $e$. Let $x$ be any point on $e$. To compute the peripheral center on $e$, we first determine function $y=\Ed(P_i,x)$ for each $P_i\in\calP$. Removing $e$ but keeping $u$ and $v$ generates two disjoint subtrees of $T$. Let $T_1$ be the one containing $u$ and $T_2$ be the other. Clearly, $T_1$ and $T_2$ can be computed in $O(mn)$ time. We then traverse $T_1$ to compute the probability sum of each $P_i$ in $O(|T_1|+n)$ time as follows. First, create an auxiliary array $F[1\cdots n]$ and initialize it as zero. We then traverse $T_1$: For each location $p_{ij}$ on $T_1$, we add $f_{ij}$ to $F[i]$. We proceed with computing $\Ed(P_i,u)$ for each $P_i\in\calP$ in $O(mn)$ time. For each $P_i$, we have $\Ed(P_i,x) = \Ed(P_i,u)+ w_iF[i]\cdot d(u,x)- w_i(1-F[i])\cdot d(u,x)$, which is a linear function in $d(u,x)$, and it can be determined in constant time. 

Let $\calP'$ be the subset containing all uncertain points whose medians are on $T_1$, i.e., whose probability sums in $T_1$ are at least $0.5$. So, $\calP/\calP'$ consists of all uncertain points whose medians are in $T_2$. For each $P_i$ in $\calP'$ (resp., $\calP/\calP'$), $\Ed(P_i,x)$ is monotonically increasing (resp., decreasing) as $x$ moves from $u$ to $v$. It is not hard to see that the peripheral center on $e$ is decided by $\calP'$ or $\calP/\calP'$. Specifically, if it is decided by $\calP'$, then this peripheral center must be placed at the point where the maximum expected distance of uncertain points in $\calP'$ is exactly $\lambda$. 
% in order to cover as many uncertain points as possible, it must be placed at the point at which the maximum expected distance of uncertain points of $\calP'$ is exactly $\lambda$. 
We say that such a point is a \textit{critical} point of $\calP'$ on $e$. Otherwise, the peripheral center must be at the point on $e$ where the maximum expected distance of uncertain points in $\calP/\calP'$ is exactly $\lambda$. 

To decide this peripheral center, it suffices to compute the critical points of $\calP'$ and $\calP/\calP'$ on $e$. For each $P_i\in\calP'$, we resolve $\Ed(P_i,x) = \lambda$ in $O(1)$ time, and it generates a value $t_i$ that defines a point $x_{t_i}$ at distance $t_i$ to $u$ with $\Ed(P_i,x_{t_i}) = \lambda$. Clearly, if $t_i<0$ then $\Ed(P_i,x)>\lambda$ for any $x\in e$; if $t_i>d(u,v)$ then $\Ed(P_i,x)<\lambda$ for any $x\in e$; otherwise, there exist a point on $e$ with $\Ed(P_i,x)=\lambda$. We also resolve $\Ed(P_i,x) = \lambda$ for each $P_i\in(\calP/\calP')$ in $O(n)$ time. 

Let $t'$ (resp., $t''$) be the smallest (resp., largest) value among these obtained values for $\calP'$ (resp., $\calP/\calP'$), which can be obtained in $O(n)$ time. If $t'<0$ or $t'>d(u,v)$ then $\calP'$ has no critical point on $e$ and hence the peripheral center on $e$ is decided by $\calP/\calP'$. This implies that this peripheral center is on $e$ at distance $t''$ to $u$. We then compute the expected distances of all uncertain points to this center and set the weight as zero of every uncertain point whose expected distance at it is at most $\lambda$. This can be carried out in $O(mn)$ time. Next, we solve the one-center problem for $\calP$ on $T$ in $O(mn)$ time by the algorithm~\cite{ref:WangCo17}, which returns the optimal objective value $\lambda'$ of the one-center problem of $\calP$ on $T$ as well as their center. Clearly, if $\lambda'\leq\lambda$ then there exists another point on $T$ that can cover all uncertain points under $\lambda$ with weights larger than zero. This means that two points can be found on $T$ to cover $\calP$ under $\lambda$. Thus, $\lambda$ is feasible, and we return that peripheral center and the center returned by~\cite{ref:WangCo17}. Otherwise, at least two centers are needed to cover all remaining uncertain points and so we return $\lambda<\lambda^*$. 

If $t''<0$ or $t''>d(u,v)$ then $\calP/\calP'$ has no critical point on $e$ and hence the peripheral center on $e$ is decided by $\calP'$, i.e., it is on $e$ at distance $t'$ to $u$. Similarly, we can decide in $O(mn)$ time the feasibility of $\lambda$ and the two centers if it is feasible. Otherwise, $\lambda$ is feasible in the sense that both $\calP'$ and $\calP/\calP'$ have critical points on $e$. We immediately return the two critical points as centers. 

Overall, for each above case, we can decide in $O(mn)$ time the feasibility of $\lambda$ as well as the two centers if it is feasible. Thus, the lemma holds. \qed
\end{proof}

To find a peripheral-center edge on $T$, each round of our algorithm consists of two pruning steps: The first pruning step is a recursive procedure that ``shrinks'' $T$ to find a peripheral-center edge recursively. After at most  $\log m+1$ recursive steps, we obtain a subtree of at most $|\calP|/2$ vertices. It follows that at least a quarter of uncertain points are pruned from $\calP$ in the second step. After at most $\log n$ rounds, only $O(1)$ uncertain points remain, or a subtree of $T$ with $O(1)$ vertices is obtained. At this point, a peripheral-center edge can be found in $O(mn)$ time. 

\subsection{The first pruning step}
We first compute the \textit{centroid} $c$ of $T$, which is a vertex so that every split subtree of $c$ has no more than $|T|/2$ vertices, in $O(|T|)$ time by traversing the tree~\cite{ref:KarivAn79,ref:MegiddoLi83}. We then decide whether a center is at $c$, and if yes, then Lemma~\ref{lem:basecase} is used to decide the feasibility of $\lambda$ in $O(mn)$ time. Otherwise, we determine which split subtree of $c$ contains a peripheral-center edge. Lemma~\ref{lem:point} can be utilized to solve this problem in $O(|T|)$ time. Note that if $c$ has more than one such split subtree then $c$ is associated with a flag equal to true in Lemma~\ref{lem:point}. 

Next, we set $c$ as a \textit{connector} on the obtained split subtree $T'$ in $O(|T|)$ time as~\cite{ref:WangCo17} by traversing the \textit{connector subtree} $\{c\}\cup (T/T')$, denoted by $T'(c)$, to compute the location information of each $P_i$ in $T'(c)$. Specifically, we first create two \textit{information arrays} $F_c[1\cdots n]$ and $D_c[1\cdots n]$. Initialize them as zero. We then traverse $T'(c)$ and during the traversal, for each location $p_{ij}$, we add $f_{ij}$ to $F_c[i]$ and add value $w_if_{ij}\cdot d(p_{ij},c)$ to $D_c[i]$. Last, we associate $c$ on $T'$ with the two information arrays. Clearly, for each $1\leq i\leq n$, $F_c[i]$ is the probability sum of $P_i$ in $T'(c)$, and $D_c[i]$ is called the \textit{distance sum} of $P_i$'s locations in $T'(c)$ to $c$. 

In $O(|T|)$ time, we obtain a subtree of at most $|T|/2$ vertices that must contain a peripheral-center edge. We continue to perform the above procedure recursively on $T'$ to find a peripheral-center edge. Suppose we are about to perform the $h$-th recursive step. Denote by $T_{h-1}$ the obtained subtree after the $h-1$-th recursive step. $T_{h-1}$ consists of at most $|T|/2^{h-1}$ vertices and at most $h-1$ connectors. 

Similarly, we first compute its centroid $c$ in $O(|T_{h-1}|)$ time. During the traversal, we also count the number of vertices with true flags. If more than one such vertices exist, then at least three centers must be placed on $T$ to cover $\calP$ under $\lambda$ since two (pruned) connector subtrees and $T_h$ each must contain center(s) and so we return $\lambda<\lambda^*$. Otherwise, we decide if one center must be placed at $c$. If yes, Lemma~\ref{lem:basecase} is applied to decide $\lambda$'s feasibility immediately in $O(mn)$ time. Otherwise, we determine which split subtree of $c$ contains a peripheral-center edge. Lemma~\ref{lem:point} can address the problem in $O(|T_{h-1}|+n\cdot (h-1))$ time. 

In general, we obtain a subtree $T_h$ of at most $|T|/2^h$ vertices, and it must contain a peripheral-center edge. Note that $T_h$ has $h$ connectors. We then set $c$ as a connector on $T_h$ by computing its information arrays $F_c[1\cdots n]$ and $D_c[1\cdots n]$. This can be done in the above way except that when we visit a connector $u$ on $T_h$, we scan its information arrays to add $F_u[i]$ to $F_c[i]$ and add value $D_u[i] + w_iF_u[i]d(u,c)$ to $D_c[i]$ for each $1\leq i\leq n$. Thus, the time complexity of the $h$-th recursive step is $O(|T_{h-1}|+n\cdot (h-1))$. 

We perform the above procedure for $h = 1+\log m$ recursive steps. At this moment, by the definition of the centroid, we have $|T_h|\leq |T|/2^h = mn/2^h = n/2$. 
% , we have $|T_h|\leq |T|/2^h = mn/2^h = n/2$. 
The running time of the $h$ recursive steps, i.e., the first pruning step, is $O(\sum_{i=0}^{h-1}(T_i+i\cdot n))$, which is $O(mn)$ due to $T_0=T$ and $h = 1+\log m$. 

Due to $|T_h|\leq n/2$, we can see that there are at least $n/2$ uncertain points in $\calP$ so that they have no locations in $T_h$. Let $\calP'$ be the subset of these uncertain points. To compute this subset, we create an array $I[1\cdots n]$ so that $P_i\in\calP'$ if $I[i]$ is true. It is easy to see that $I[1\cdots n]$, i.e., $\calP'$, can be computed in $O(|T_h|+n)$ time by traversing $T_h$. Below we will show in the second pruning step that at least half of $\calP'$, i.e., a quarter of uncertain points in $\calP$, can be pruned. 

\subsection{Pruning uncertain points}
We first traverse $T_h$ to count the number of vertices with true flags. If more than one such vertices exist, then we return $\lambda< \lambda^*$ in that at least three centers are needed to be placed on $T$. Otherwise, denote by $C$ the number of connectors on $T_h$. Depending on the value of $C$, our algorithm will proceed accordingly for three cases: $C=1$, $C=2$ and $C>2$. 

\subsubsection{The case $C=1$:} Let $c$ be the connector on $T_h$. If the flag of $c$ is true, then the connector subtree $T_h(c)$ must contain centers. By Lemma~\ref{lem:point}, we also have that peripheral-center edge(s) on $T_h$ must be decided by uncertain points with medians on $T_h$, i.e., whose probability sums in $T_h$ are at least $0.5$. Because Lemma~\ref{lem:point} always returns a subtree without any true-flag vertices unless no such subtrees exist. Keeping $c$'s flag being true leads that $\calP'$ can be pruned as computing a peripheral-center edge on $T_h$ further. So, we keep entries of $I[1\cdots n]$ for $\calP'$ being true. 

Otherwise, the connector subtree $T_h(c)$ does not contain any centers. Denote by $x_t$ any point on $T_h$ at distance $t$ to $c$. For each $P_i\in\calP'$, since all its locations are in $T_h(c)$, we have $\Ed(P_i,x_t) = \Ed(P_i,c) +w_i\cdot t$ where $\Ed(P_i,c) = D_c[i]$. Denote by $P_{i'}$ the uncertain point of $\calP'$ that determines the furthest point to $c$ on $T_h$ to cover $\calP'$, i.e., that has the smallest $t$-value by resolving $\Ed(P_i,c) +w_i\cdot t =\lambda$ for all $P_i\in\calP'$. Clearly, any point on $T_h$ that covers $P_{i'}$ must cover all in $\calP'$. Thus, $\calP'-\{P_{i'}\}$ can be pruned in the further step to compute a peripheral-center edge on $T_h$. To find $P_{i'}$, we scan $I[1\cdots n]$ to compute the $t$-value of each $P_i\in\calP'$: If $I[i]$ is true, we compute in $O(1)$ time $(\lambda - \Ed(P_i,c))/w_i$, i.e., the $t$-value. The smaller $t$-value is always maintained during the scan so that $P_{i'}$ can be found in $O(n)$ time. Last, we set $I[i']$ as false. Hence, it takes $O(n)$ time to find all uncertain points of $\calP'$ that will be pruned.     

It can be seen that for either case, at least a half of $\calP'$, that is, a quarter of uncertain points in $\calP$, can be pruned as we further search for peripheral-center edges on $T_h$. Additionally, if $I[i]$ for any $1\leq i\leq n$ is true then $P_i$ can be pruned. Last, we reconstruct a tree to prune these uncertain points. Traverse the connector subtree $T_h(c)$, i.e., $\{c\}\cup T/T_h$. For each location $p_{ij}$ on $T_h(c)$ with $I[i]$ being false, we create a \textit{dummy} vertex $v$, set $v$ as an adjacent vertex of $c$ on $T_h$ by a \textit{dummy} edge of length $d(c,v)$, and reassign $p_{ij}$ to $v$. These can be carried out totally in $O(|T|)$ time. 

A tree $T^+$ is thus obtained: $T^+$ contains at most $3n/4$ uncertain points and hence its size is $3mn/4$. $T^+$ has dummy vertices and each of them is a leaf. Additionally, the subtree generated by removing all dummy vertices from $T^+$ is exactly $T_h$ which contains a peripheral-center edge. As the flag of $c$ on $T^+$ is maintained, computing a peripheral-center edge on $T_h$ is equivalent to computing a non-dummy peripheral-center edge on $T^+$. 

\subsubsection{The case $C=2$:} Let $c_1$ and $c_2$ be the two connectors on $T_h$. Consider the path $\pi(c_1,c_2)$ between $c_1$ and $c_2$. Let $V$ be the set of vertices on $\pi(c_1, c_2)$ except for $c_1$ and $c_2$. For any vertex $v\in V$, denote by $\Gamma(v)$ the set of all split subtrees of $v$ in $T_h$ excluding the two containing $c_1$ and $c_2$. Let $\Gamma(V)= \cup_{v\in V}\Gamma(v)$. $V$ and $\Gamma(V)$ can be computed in $O(T_h)$ time. See Fig.~\ref{fig:fig2} for an example.

We proceed with determining whether $\Gamma(V)$ contains peripheral-center edges. This is an instance of the \textit{center-edge detecting} problem, defined below in Section~\ref{sec:lemmas}, and it can be solved in $O(|T|)$ time by Lemma~\ref{lem:centerdetecting}. If one split subtree $T'$ of $\Gamma(V)$ is returned then we can reduce this case to the case $C=1$ by setting the vertex $v'$ of $T'$ in $V$ as a connector on $T'$ in $O(|T_h|+n)$ time. Note that the flag of $v'$ is set properly in Lemma~\ref{lem:centerdetecting}. 
% We then apply Lemma~\ref{lem:point} compute in $O(|T_h|+n)$ time the flag of $v'$. 
% Notice that if some vertex of $T_h/T'$ has a true flag then the flag of $v'$ is set as true in Lemma~\ref{lem:centerdetecting}. 
Otherwise, if  no subtrees are returned then $\pi(c_1,c_2)$ must contain a peripheral-center edge. Depending on whether $T_h(c_1)$ and $T_h(c_2)$ contain centers, we have different pruning approaches. 

Suppose $c_1$'s flag is equal to true, i.e., $T_h(c_1)$ contains centers. Let $\calP'_{c_1}$ be the subset of uncertain points in $\calP'$ whose medians are in $T_h(c_1)$ and $\calP'_{c_2}$ be the subset of the remaining uncertain points in $\calP'$ whose medians are in $T_h(c_2)$. Lemma~\ref{lem:point} implies that the peripheral-center edge on $\pi(c_1, c_2)$, i.e., that center, is decided by uncertain points whose medians are in $T/T_h(c_1)$. Keeping $c_1$'s flag being true allows us to prune $\calP'_{c_1}$ in the further rounds of computing a peripheral-center edge on $\pi(c_1, c_2)$. So, we keep entries of $I[1\cdots n]$ for $\calP'_{c_1}$ being true. 

Next, we find uncertain points of $\calP'_{c_2}$ that can be pruned in the further rounds. Let $x_t$ be any point on $\pi(c_1,c_2)$ at distance $t$ to $c_1$. Because all locations of uncertain points in $\calP'$ are in the two connector subtrees. Each $P_i\in\calP'$ has $\Ed(P_i,x_t)=\Ed(P_i,c_1) +w_iF_{c_1}[i]\cdot t - w_iF_{c_2}[i]\cdot t$, which is $D_{c_1}[i] + D_{c_2}[i] + w_i(F_{c_1}[i]-F_{c_2}[i])t +w_iF_{c_2}[i]d(c_1, c_2)$. Clearly, $\Ed(P_i,x_t)$ changes linearly as $x_t$ moves from $c_1$ to $c_2$. We resolve $\Ed(P_i,x_t) = \lambda$ in $O(1)$ time for each $P_i\in\calP'_{c_2}$. (For any $1\leq i\leq n$, if $I[i]$ is true and $F_{c_2}[i]\geq 0.5$ then $P_i\in\calP'_{c_2}$.) Let $t'$ be the largest value among all obtained, and $t'$ defines the furthest point to $c_2$ on $\pi(c_1,c_2)$ that covers $\calP'_{c_2}$. 

If $t'< 0$ then $\Ed(P_i,x)<\lambda$ for each $P_i\in\calP'_{c_2}$ and any $x\in\pi(c_1,c_2)$. This means that any center on $\pi(c_1, c_2)$ is irrelevant to $\calP'_{c_2}$. We thus are allowed to prune $\calP'_{c_2}$ when we shall compute a peripheral-center edge on $\pi(c_1, c_2)$. Otherwise, $t'\geq 0$ and let $x_{t'}$ be the point at distance $t'$ to $c_1$ on $\pi(c_1,c_2)$. We then call Lemma~\ref{lem:point} on $x_{t'}$ to decide in $O(|T_h|+n)$ time whether $x_{t'}$ must contain a center. If yes, then the feasibility of $\lambda$ can be decided in $O(mn)$ time by Lemma~\ref{lem:basecase}, and otherwise, a peripheral-center edge must be on $\pi(x_{t'}, c_2)$ in that the split subtrees in $\Gamma(V)$ of each $v$ on $\pi(x_{t'}, c_2)$ do not contain centers. It follows that $\calP'_{c_2}$ can be pruned during the further searching for the peripheral-center edge on $\pi(x_{t'}, c_2)$. Hence, we join a vertex $v'$ for $x_{t'}$ into the path if $x_{t'}$ is in the interior of an edge, and set the flag of $v'$ as true in $O(1)$ time to find a peripheral-center edge on $\pi(v', c_2)$ in further steps. For either $t'<0$ or $t'\geq 0$, $\calP'_{c_2}$ can be pruned, and so we keep entries of $I[1\cdots n]$ being true for $\calP'_{c_2}$ in order to prune them further. 

% It follows that $\calP'_{c_2}$ can be pruned as we further seek for the peripheral-center edge on $\pi(v', c_2)$. So, we keep $I[1\cdots n]$ as the above to prune $\calP'_{c_2}$. 

% Thus, we keep entries of $I[1\cdots n]$ being true for $\calP'_{c_2}$. 

It is clear to see that if one of $c_1$ and $c_2$ has a true flag, we find in $O(|T|)$ time a subpath $\pi$ of $\pi(c_1,c_2)$ that contains a peripheral-center edge and at least a half of uncertain points in $\calP'$ that can be pruned in the further rounds of computing a peripheral-center edge on $\pi$. Notice that the flags of vertices are maintained in the above procedure.   

On the other hand, both $c_1$ and $c_2$ have false flags. We first compute the furthest point $x_1$ (resp., $x_2$) to $c_1$ (resp., $c_2$) on $\pi(c_1, c_2)$ that covers $\calP'_{c_1}$ (resp., $\calP'_{c_2}$), which can be computed in $O(n)$ time as the above. We then find irrelevant uncertain points for each below case. 

In the first case, $x_1$ is to the left of $x_2$. There must be a peripheral-center edge on $\pi(x_2, c_2)$ (resp., $\pi(c_1, x_1)$) if split subtrees in $\Gamma(V)$ of vertices on $\pi(x_2, c_2)$ (resp., $\pi(c_1, x_1)$) do not contain any vertex with a true flag, which can be determined in $O(|T_h|)$ time. Suppose it is on $\pi(x_2, c_2)$. Since any point on $\pi(x_2, c_2)$ cannot cover $\calP'_{c_1}$ under $\lambda$, $\calP'_{c_1}$ can be pruned by setting the flag of $x_2$ as true in our further searching on $\pi(x_2, c_2)$. Additionally, all in $\calP'_{c_2}$ except for the one determining $x_2$ can be pruned, which can be found in $O(n)$ time. 
% as we computing a peripheral-center edge on $\pi(x_2,c_2)$. 
So, we join a vertex $v'$ of a true flag for $x_2$ if necessary, and set the entry in $I[1\cdots n]$ as false for the uncertain point determining $x_2$. The case where $\pi(c_1, x_1)$ does contain a peripheral-center edge can be processed similarly. It is not hard to see that in $O(|T_h|+n)$ time we obtain a subpath $\pi$ of $\pi(c_1,c_2)$ that must contain a peripheral-center edge, and all irrelevant uncertain points can be found in $O(n)$ time. 

% If $\pi(c_1, c_1)$ does contain a peripheral-center edge, then $\calP'_{c_2}$ and all in $\calP'_{c_1}$ except for the one determining $x_1$ can be pruned as computing a peripheral-center edge on $\pi(c_1,x_1)$ further. 

Another case is that $x_1$ is to the right of $x_2$. Similarly, if split subtrees in $\Gamma(V)$ of vertices on $\pi(c_1, x_2)$ (resp., $\pi(x_1, c_2)$) contain any vertex of a true flag then there must be a peripheral-center edge on $\pi(x_2, c_2)$ (resp., $\pi(c_1, x_1)$), which can be known in $O(|T_h|)$ time. If $\pi(x_2, c_2)$ (resp., $\pi(c_1, x_1)$) contains a peripheral-center edge, then $\calP'$ can be pruned except for the one in $\calP'_{c_2}$ (resp., $\calP'_{c_1}$) determining $x_2$ (resp., $x_1$). For the situation where a peripheral-center edge is on $\pi(x_2, c_2)$ (resp., $\pi(c_1, x_1)$), we join a vertex for $x_2$ if necessary and set its flag as true, and then reset the entry of $I[1\cdots i]$ as false for that uncertain point determines $x_2$ (resp., $x_1$). Clearly, all these operations can be done in $O(|T_h|+n)$ time. 

Otherwise, only split subtrees in $\Gamma(V)$ of vertices on $\pi(x_2, x_1)/\{x_1, x_2\}$ may contain a true-flag vertex. If such a split subtree exists, supposing it intersects $\pi(c_1, c_2)$ at vertex $v'$, then we apply Lemma~\ref{lem:point} to $v'$ to decide which split subtree of $v'$ contains a peripheral-center edge. Because $\pi(c_1,c_2)$ must contain a peripheral-center edge. A split subtree of $v'$ containing $c_1$ or $c_2$ must be returned. If it includes $c_2$ (resp., $c_1$), then $\calP'$ can be pruned as we further compute a peripheral-center edge on $\pi(v',c_2)$ (resp., $\pi(c_1,v')$). For either case, we maintain the path containing a peripheral-center edge and set the flag of $v'$ as true. Clearly, this case can be handled in $O(|T_h|+n)$ time as well.  

If no true-flag vertices are found in split subtrees of $\Gamma(v)$ of vertices on $\pi(x_2, x_1)/\{x_1, x_2\}$, then solving the decision problem is equivalent to solving the \textit{path-constrained} version that is to decide whether two points can be found on $\pi(c_1, c_2)$ to cover $\calP$ by the given $\lambda$. This is because no centers are necessary to be placed on split subtrees of $\Gamma(V)$ and the two connectors subtrees. Moreover, the path-constrained decision version can be solved in $O(mn)$ time by Lemma~\ref{lem:pathconstrained}. 

\begin{lemma}\label{lem:pathconstrained}
The path-constrained decision problem can be solved in $O(mn)$ time. 
\end{lemma}
\begin{proof}
In the path-constrained version, all centers are restricted to be on a given path $\pi(u,u')$ of $T$. Denote by $V$ the set of vertices on $\pi(u,u')$ including $u, u'$. Let $\Gamma(V)$ be the set of all split subtrees of vertices in $V$ excluding those containing $u$ and $u'$. $V$ and $\Gamma(V)$ can be found in $O(|T|)$ time. 

Let $x$ be any point on $\pi(u,u')$. Because $\Ed(P_i,x)$ of each $P_i\in\calP$ is a convex and piece-wise linear function of $O(m)$ complexity. $\Ed(P_i,x)\leq\lambda$ leads an interval $I_i$ on $\pi(u,u')$. We say that an interval is pierced by a point if it contains this point. It is evident that the path-constrained decision problem is to determine whether at most two points can be found on $\pi(u,u')$ to pierce all $I_i$'s. If yes then $\lambda\geq\lambda^*$ and otherwise, $\lambda<\lambda^*$. As revealed in~\cite{ref:XuTh23} for the case of $T$ being a path, all $I_i$'s can be pieced by at most two points if and only if their leftmost right endpoint and rightmost left endpoint pierce all of them; clearly, this can be verified in $O(mn)$ time if all $I_i$'s are known. 

It thus remains to compute all $I_i$'s, which require us to determine $\Ed(P_i,x)$ for each $P_i\in\calP$. Create an array $D[1\cdots n]$ to store the distance sum of locations in $\Gamma(V)$ of each $P_i$ to vertices in $V$. Traverse every split subtree in $\Gamma(V)$ of each $v\in V$. During the traversal, for each location $p_{ij}$, we add $w_if_{ij}d(p_{ij}, v)$ to $D[i]$ and then reassign $p_{ij}$ to $v$. Let $F_i(<x)$ (resp., $F_i(>x)$) be the set of $P_i$'s locations on $\pi(u,u')$ to the left (resp., right) of $x$. Clearly, $\Ed(P_i,x) = D[i]+ w_i\sum_{p_{ij}\in F_i(<x)}f_{ij}d(p_{ij},x) + w_i\sum_{p_{ij}\in F_i(<x)}f_{ij}d(p_{ij},x)$. 
Since all locations of $\calP$ are at vertices of path $\pi(u,u')$ and thus sorted, we can utilize the approach~\cite{ref:XuTh23} for the path version to determine all $\Ed(P_i,x)$'s and then resolve all $\Ed(P_i,x)\leq\lambda$ in $O(mn)$ time. 

Thus, the lemma holds. \qed
\end{proof}
% \begin{proof}
% In the path-constrained version, all centers are restricted to be on a given path $\pi(u,u')$ of $T$. Denote by $V$ the set of vertices on $\pi(u,u')$ including $u, u'$. Let $\Gamma(V)$ be the set of all split subtrees of vertices in $V$ excluding those containing $u$ and $u'$. $V$ and $\Gamma(V)$ can be found in $O(|T|)$ time. Let $x$ be any point on $\pi(u,u')$. To place centers on $\pi(u,u')$, we first determine $\Ed(P_i,x)$ for each $P_i\calP$ as the following. 

% Create an array $D[1\cdots n]$ to store the distance sums of locations in $\Gamma(V)$. Traverse each split subtree in $\Gamma(V)$ for each $v\in V$: For each location $p_{ij}$, we add $w_if_{ij}d(p_{ij}, v)$ to $D[i]$ and reassign $p_{ij}$ to $v$. We thus obtain a path of at most $mn$ vertices so that each $P_i\in\calP$ has $m$ locations on $T'$. 

% Let $F_i(<x)$ (resp., $F_i(>x)$) be the set of $P_i$'s locations on $\pi(u,u')$ to the left (resp., right) of $x$. Clearly, $\Ed(P_i,x) = D[i]+ w_i\sum_{p_{ij}\in F_i(<x)}f_{ij}d(p_{ij},x) + w_i\sum_{p_{ij}\in F_i(<x)}f_{ij}d(p_{ij},x)$. Now we can see that solving the path-constrained decision problem is equivalent to solving the one-dimensional two-center decision problem of $\calP$ on $\pi(u,u')$. Since all locations are given sorted on $\pi(u,u')$, we can apply the decision algorithm~\cite{ref:XuTh23} for the one-dimensional two-center decision problem to solve the path-constrained decision problem in $O(mn)$ time. Thus, the lemma holds. \qed
% \end{proof}

It follows that for the case where $x_1$ is to the right of $x_2$, we spend $O(|T_h|+n)$ time on computing a subpath $\pi$ of $\pi(c_1, c_2)$ that contains a peripheral-center edge, and at least a quarter of uncertain points of $\calP$ that can be pruned in our further rounds on $\pi$. 

Last, $x_1$ is exactly $x_2$. We first decide whether a center must be placed at $x_1$ by Lemma~\ref{lem:point} in $O(|T_h|+n)$ time. If yes, then the feasibility of $\lambda$ can be decided by Lemma~\ref{lem:basecase} in $O(mn)$ time. Otherwise, one split subtree of $x_1$ is returned and it must contain either $c_1$ or $c_2$. Assume it contains $c_2$. We shall compute a peripheral-center edge on $\pi(x_2, c_2)$ so that all uncertain points in $\calP'$ except for the one determining $x_2$ can be pruned. Similar to the above, we join a vertex $v'$ for $x_2$, set its flag as true since there must be a center on the connector subtree $T_h(v')$, and reset the entry in $I[1\cdots n]$ of that uncertain point determining $x_2$ as true.

As a consequence, for the case $C=2$, generally, in $O(|T|)$ time, either we reduce the problem into the case $C=1$, or we obtain a subpath $\pi$ on $\pi(c_1,c_2)$ that is known to contain a peripheral-center edge. For the former case, the above approach for the case $C=1$ is applied to prune at least a quarter of uncertain points of $\calP$ in $O(|T|)$ time. For the other case, supposing the two ending vertices of $\pi$ are $u$ and $u'$, we then prune these irrelevant uncertain points as follows. Traverse the connector subtrees $T_h(u)$ and $T_h(u')$, which can be obtained in $O(|T|)$ time. For each location $p_{ij}$ in $T(u)$ (resp., $T(u')$), we create a dummy vertex incident to $u$ (resp., $u'$) by a dummy edge of length $d(p_{ij},u)$ (resp., $d(p_{ij},u')$) if $I[i]$ is false. We perform the same procedure on each split subtree in $\Gamma(V)$ of each vertex $v$ on $\pi(u,u')/\{u,u'\}$ and additionally, we set $v$'s flag as true if its split subtrees in $\Gamma(V)$ have a true-flag vertex. It is not hard to see that in $O(|T|)$ time we obtain a tree $T^+$ so that $T^+$ contains at most $3n/4$ uncertain points and $T^+/\pi(u,u')$ are induced by dummy vertices. Clearly, computing a peripheral-center edge on $T$ is equivalent to computing a non-dummy peripheral-center edge on $T^+$. 

\subsubsection{The case $C>2$:} In general, $T_h$ has more than two connectors. We utilize an approach similar to~\cite{ref:WangCo17} to ``shrink'' $T_h$ until the problem is reduced to one of the previous two cases. 

A vertex $z$ of $T_h$ is called a \textit{connector-centroid} if each split subtree of $z$ has no more than $C/2$ connectors, which can be found in linear time~\cite{ref:WangCo17}. We first compute the connector-centroid $z$ of $T_h$ and then determine whether a center must be at $z$, and if not, which split subtree of $z$ contains a peripheral-center edge. These can be decided in $O(|T_h|+h\cdot n)$ time by applying Lemma~\ref{lem:point} to $z$. Generally, a split subtree is returned. Set $z$ as a connector on this subtree in $O(|T_h|+h\cdot n)$ time. Let $T_{h+1}$ be the obtained tree. Clearly, the size of $T_{h+1}$ is at most $|T_h|$ but it has no more than $C/2$ connectors. We perform the above procedure recursively on $T_{h+1}$. After at most $\log C$ steps, we obtain a subtree $T'$ with at most two connectors. The total time complexity is $O(\sum_{i=1}^{\log C}(|T_h|+n\cdot C/2^i))$, which is $O(|T|)$ due to $C\leq 1+\log m$ and $T_h\leq\frac{n}{2}$. It follows that the above $O(|T|)$ pruning approach is applied to $T'$ accordingly to prune at least a quarter of uncertain points from $\calP$. Consequently, a tree $T^+$ containing at most $3n/4$ uncertain points is achieved in $O(|T|)$ time so that computing a peripheral-center edge on $T_h$ is equivalent to computing a non-dummy peripheral-center edge on $T^+$.  

\subsection{Wrapping things up}
The above procedure gives an $O(|T|)$-time algorithm that computes a tree $T^+$ of at most $3n/4$ uncertain points and at most $3mn/4$ vertices, such that computing a peripheral-center edge on $T$ is equivalent to computing a non-dummy peripheral-center edge on $T^+$. Note that all dummy vertices on $T^{+}$ are leaves and flags of vertices are maintained. Because our Lemma~\ref{lem:point} and Lemma~\ref{lem:centerdetecting} are defined in Section~\ref{sec:lemmas} consider the situation where the given tree may contain dummy vertices and connectors. We thus continue the same procedure recursively on $T^+$ to compute a non-dummy peripheral-center edge. 

After $h-1$ rounds, the obtained tree $T^+_{h-1}$ consists of at most $(\frac{3}{4})^{h-1}\cdot n$ uncertain points and at most $(\frac{3}{4})^{h-1}\cdot mn$ vertices. It is not hard to see that performing the $h$-th round on $T^+_{h-1}$ takes $O(|T^+_{h-1}|)$ time. We stop until we obtain a tree $T'$ that contains $O(1)$ uncertain points or $T'$ consists of $O(1)$ non-dummy vertices, that is, after at most $\log n$ rounds. At this moment, a peripheral-center edge of $T$ can be computed in $O(mn)$ time by Lemma~\ref{lem:constant}. Thus, the total time complexity is $O(\sum_{i=1}^{\log n} (\frac{3}{4})^i\cdot mn )$ time, which is $O(mn)$. 

\begin{lemma}\label{lem:constant}
The non-dummy peripheral-center edge on $T'$ can be computed in $O(mn)$ time. 
\end{lemma}
\begin{proof}
On the one hand, $T'$ contains $O(1)$ uncertain points, and so $|T'| = O(m)$. Similar to the first pruning step, we compute the centroid $c$ of $T'$ and then apply Lemma~\ref{lem:point} to $c$, which can be done in $O(m)$ time. If a center must be placed at $c$, then we can determine the feasibility of $\lambda$ by Lemma~\ref{lem:basecase} in $O(mn)$ time. Otherwise, a split subtree of $c$ is returned that must contain a peripheral-center edge, and we set $c$ as a connector. Note that the size of the information arrays of a connector is $O(1)$. We recursively perform the procedure on the obtained subtree until an edge is obtained. This edge must be a non-dummy edge and it contains a peripheral center. So we return this edge. The total time complexity is $O(m)$. 

On the other hand, $T'$ contains $O(n)$ uncertain points, but $T'$ consists of $O(1)$ non-dummy vertices. We first compute the subtree $T''$ that is generated by removing all dummy vertices from $T'$ in $O(|T'|)$ time. We then perform the procedure similarly to the above on $T''$ to find a non-dummy peripheral-center edge on $T''$. Notice that during the procedure, Lemma~\ref{lem:point} is applied to $T$ instead of $T'$ so that every recursive step, we can avoid to set connectors. Consequently, after $O(1)$ recursive steps, a peripheral-center edge of $T''$ is achieved, which is also a peripheral-center edge of $T$. Clearly, the total running time is $O(mn)$ time. 

Thus, the lemma holds. \qed
\end{proof}

Once a peripheral-center edge is found, we call Lemma~\ref{lem:basecase} to decide the feasibility of $\lambda$ in an additional $O(mn)$ time. We thus have the following result. 

\begin{lemma}
The decision two-center problem can be solved in $O(mn)$ time. 
\end{lemma}

\subsection{Lemma~\ref{lem:point} and Lemma~\ref{lem:centerdetecting}}\label{sec:lemmas}
Let $T_h$ be a tree obtained after several recursive steps of the first pruning step in a round of our algorithm. Suppose $T_h$ contains $n_h$ uncertain points and $t$ connectors. Additionally, $T'$ may contain dummy vertices but all of them are leaves, and at most one vertex of $T'$ has a true flag. We have the following lemma. 

\begin{lemma}\label{lem:point}
Given any point $x$ on $T_h$, we can decide in $O(|T_h|+ t\cdot n_h)$ time whether a center must be placed at $x$, and if not, which split subtree of $x$ on $T_h$ contains a peripheral-center edge. 
\end{lemma}
\begin{proof}
If $x$ is in the interior of any dummy edge in $T_h$ or it is a dummy vertex, which can be known in $O(1)$ time, then only the split subtree containing non-dummy edges has peripheral-center edges. (Note that any point $x$ interior of a dummy edge has only one split subtree containing non-dummy edges while all other split subtrees contain only dummy edges, and so is a dummy vertex.) This split subtree can be found in $O(|T_h|)$ time. 

In general, $x$ is in the interior of neither any non-dummy edge nor a dummy vertex. Let $T^1_h, \cdots, T^s_h$ be the split subtrees of $x$ on $T_h$. Let $\calP_k$ be the subset of uncertain points whose medians are interior of $T^k_h/\{x\}$. By Lemma~\ref{lem:medians}, we have that each $P_i\in\calP_k$ has its probability sum at least $0.5$ in $T^k_h$ but less than $0.5$ in every other split subtree. So, subsets $\calP_k$ for all $1\leq k\leq s$ are disjoint. Let $\calP^+_k$ be the subset of uncertain points so that each $P_i\in\calP^+_k$ has its probability sum in both $T^k_h$ and another split subtree of $x$ exactly equal to $0.5$. Define $\tau(\calP_k,x)=\max_{P_i\in\calP_k}w_i\Ed(P_i,x)$. 

The convexity of $\Ed(P_i,x)$ for any path leads the following observations. For any $T^k_h$ without any vertex of a true flag, if $\tau(\calP_k,x)>\lambda$ then $T^k_h$ must contain a peripheral-center edge; if $\tau(\calP_k,x)=\lambda$ and $\tau(\calP^+_k,x)\leq\lambda$ then one must place a peripheral center at $x$ to cover $\calP_k\cup\calP^+_k$; if $\tau(\calP^+_k,x)>\lambda$, due to $\lambda^*\geq\max_{1\leq i\leq n}\Ed(P_i,p^*_i)$, then we return $\lambda<\lambda^*$; otherwise, no peripheral-center edges are in $T^k_h$. Because $x$ has at most one split subtree that consists of a true-flag vertex. The split subtree with a true-flag vertex must have a peripheral-center edge if and only if every other split subtree does not contain peripheral-center edges. Additionally, if some uncertain point with its unique median at $x$ has its expected distance to $x$ equal to $\lambda$, then a center must be placed at $x$. By Lemma~\ref{lem:medians}, if an uncertain point has its unique median at $x$ then its probability sum on each split subtree of $x$ is less than $0.5$. Denote by $\calP''$ the subset of all such uncertain points. Clearly, $\calP''$ is $\calP/\cup_{k=1}{s}(\calP_k\cup\calP^+_k)$. We thus first compute $\tau(\calP_k,x)$ and $\tau(\calP^+_k,x)$ for each $1\leq k\leq s$ and $\tau(\calP'',x)$.  

Create four auxiliary arrays $L[1\cdots n_h]$, $R[1\cdots n_h]$, $A[1\cdots n_h]$, and $B[1\cdots n_h]$: $L[i]$ and $R[i]$ are used to determine which of $\calP_k$'s and $\calP^+_k$'s $P_i$ is in; $A[i]$ stores the probability sum of $P_i$ in the current split subtree; $B[i]$ stores $w_i\Ed(P_i,x)$. Initially, $L[1\cdots n_h]$ and $R[1\cdots n_h]$ are set as null, and $A[1\cdots n_h]$ and $B[1\cdots n_h]$ are all zero. 

Traverse each split subtree of $x$. During the traversal on each $T^k_h$, when we visit a vertex holding a location $p_{ij}$, we add $f_{ij}$ to $A[i]$ and $w_if_{ij}\cdot d(p_{ij},x)$ to $B[i]$, and we then check if $A[i]>=0.5$. If yes, then $P_i$ belongs to either $\calP_k$ or $\calP^+_k$, and we continue to check if $L[i] = null$. If so, we set $L[i] = T^k_h$, and otherwise, supposing $L[i] = T^j_h$, $P_i$ belongs to both $\calP^+_j$ and $\calP^+_k$ and thereby we let $R[i] = T^k_h$. All of these steps can be done in constant time. As we visit a connector $u$ of $T_h$, we scan its information arrays $F_u[1\cdot n_h]$ and $D_u[1\cdot n_h]$: For each $1\leq i\leq n_h$, we add $F_u[i]$ to $A[i]$ and add $D_u[i] + w_iF_u[i]\cdot d(u,x)$ to $B[i]$. We then decide which subset $P_i$ belongs to as the above in $O(1)$ time. Once we are done, we revisit $T^k_h$ again: For each location $p_{ij}$ at a vertex, we reset $A[i]=0$; for each connector $v$, we scan $F_v[1\cdot n_h]$ to reset $A[i]=0$ if $F_v[i]>0$. Clearly, it takes $O(|T_h|+ t\cdot n_h)$ time in total to process all split subtrees. 

After the above procedure, we have $\Ed(P_i,x) = B[i]$ for each $1\leq i\leq n_h$. Additionally, if $L[i]=null$ then $P_i$ has its unique median at $x$; if $L[i]=T^k_h$ and $R[i] = null$ then $P_i$ belongs to $\calP_k$; otherwise, supposing $L[i]=T^k_h$ and $R[i] = T^j_h$, $P_i$ belongs to both $\calP^+_j$ and $\calP^+_k$. 

We proceed to compute $\tau(\calP_k,x)$ and $\tau(\calP^+_k,x)$ for each $1\leq k\leq s$. For each $1\leq i\leq n_h$, if $L[i] = T^k_h$ and $R[i] = null$, we set $\tau(\calP_k,x)$ as $B[i]$ if $B[i]$ is larger, and if $L[i] = T^k_h$ and $R[i] = T^j_h$, we set both $\tau(\calP^+_k,x)$ and $\tau(\calP^+_j,x)$ as $B[i]$ if $B[i]$ is larger; if $L[i] = null$ then $P_i$ belongs to $\calP''$ and we set $\tau(\calP'',x)$ as $B[i]$ if $B[i]$ is larger. Clearly, this can be carried out in $O(n_h)$ time.  

Now we are ready to determine whether there must be a center at $x$, and if not, which split subtree contains peripheral-center edges. Create a list $\Gamma$ to store the split subtrees containing peripheral-center edges. Traverse each split subtree to find that split subtree containing a true-flag vertex in $O(|T_h|)$ time, and add it to $\Gamma$. Assume it is $T^j_h$. We proceed with checking if $\tau(\calP'',x)= \lambda$. If yes, then we return $x$ as a center, and otherwise, we return $\lambda$ is infeasible if $\tau(\calP'',x)>\lambda$. Next, we decide whether $\tau(\calP^+_k,x)>\lambda$ for any $1\leq k\leq s$, and if one exists then $\lambda$ is infeasible. Next, for each $1\leq k\neq j\leq s$, we test if $\tau(\calP_k,x)>\lambda$. If yes, then we add $T^k_h$ into the list $\Gamma$ in that it must contain a peripheral-center edge, and otherwise, if $\tau(\calP_k,x)=\lambda$, then a peripheral center must be placed at $x$ and thereby $x$ is returned. Last, we scan $L[1\cdot n_h]$ and $R[1\cdot n_h]$ to do the following: if $L[i]= T^k_h$ and $R[i]=T^j_h$, we test whether $B[i]=\lambda$ and if yes, we return $x$ as a center in the case of $\tau(\calP_k,x)<\lambda$ and $\tau(\calP_j,x)\lambda$. Clearly, all these can be carried out in $O(|T_h|+n_h)$ time. 

In general, we obtain a list $\Gamma$ of $x$'s split subtrees that must contain centers. Depending on the size of $\Gamma$, there are several cases to consider. If its size is larger than two, then $\lambda$ is infeasible and we are done. If its size is equal to two, then at least one of them contains no vertex with a true flag. We thus return that split subtree, which can be found in $O(|T_h|)$ time. Notice that in this situation we also set the flag of $x$ as true in that it has two split subtrees containing centers. If $x$ is not a vertex of $T$, we join $x$ to $T_h$ and set its flag as true, which can be done in constant time. If $\Gamma$ consists of only one split subtree then we return this split subtree immediately. Otherwise, $\lambda$ is much larger than $\lambda^*$ and so we return $x$ as a center. 

Overall, in $O(|T_h|+ t\cdot n_h)$ time, we can decide whether a center must be placed at $x$, and if not, which split subtree of $x$ contains a peripheral-center edge. Thus, the lemma holds. \qed  
\end{proof} 

Furthermore, we have the following result for solving the peripheral-center detecting problem on the obtained tree $T^+$ after several rounds of the above algorithm. So, $T^+$ has $n^+$ uncertain points and $|T^+| = mn^+$. All dummy vertices of $T^+$ are leaves, and $T^+$ may contain at most one true-flag vertex. Note that no connectors are on $T^+$. 

Given are a set $Y$ of points $y_1, \cdots, y_t$ on $T^+$ and a set $\Gamma$ of split subtrees $T^+_1, \cdots, T^+_s$ of points in $Y$ on $T^+$. The peripheral-center detecting problem is to decide which split subtree of $\Gamma$ contains a peripheral-center edge. 

\begin{lemma}\label{lem:centerdetecting}
The peripheral-center detecting problem can be solved in $O(|T^+|)$ time. 
\end{lemma}
\begin{proof}
Let $Y(T^+_k)$ be the point in $Y\cap T^+_k$. Recall that $\calP_k$ is the subset of uncertain points that have their probability sums in $T^+_k/Y(T^+_k)$ at least $0.5$. For each split subtree $T^+_k$, we will compute the value $\tau(\calP_k,Y(T^+_k))$ which is $\max_{P_i\in\calP_k}w_i\Ed(P_i,Y(T^+_k))$. For a split subtree without any true-flag vertex, if $\tau(\calP_k,Y(T^+_k))>\lambda$ then $T^+_k$ must contain a peripheral-center edge. In order to solve this problem, it suffices to compute $\tau(\calP_k,Y(T^+_k))$ for each $1\leq k\leq s$. 

\begin{figure}[h]
\centering
 \includegraphics[width=0.6\textwidth]{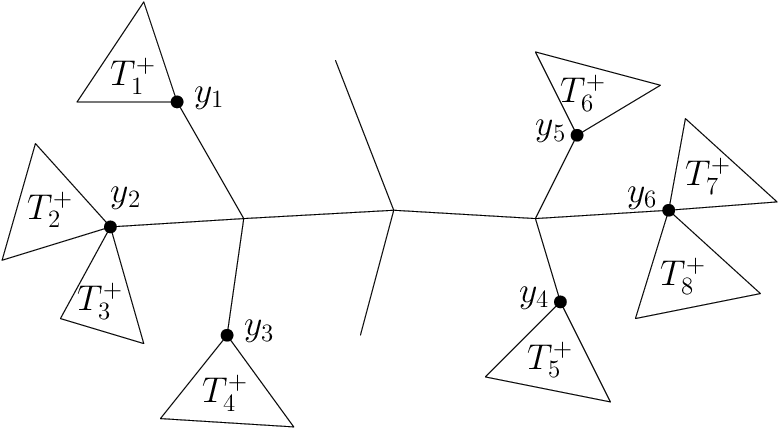}
 \caption{Illustrating an example for the peripheral-center detecting problem: $Y=\{y_1, y_2, \cdots, y_6\}$ and $\Gamma = \{T^+_1, T^+_2, \cdots, T^+_8\}$ shown with triangles.}
 \label{fig:edgedetecting}
\end{figure}

We first apply the lowest common ancestor data structure~\cite{ref:BenderTh00} to $T^+$ in $O(|T^+|)$ time so that for any two points $p$ and $q$ on $T^+$, distance $d(p,q)$ can be computed in constant time. We then traverse each split subtree of $\Gamma$ to compute $\calP_1, \cdots, \calP_s$ in the following. 

Create three auxiliary arrays $L[1\cdots n^+]$, $A[1\cdots n^+]$ and $B[1\cdots n^+]$. As in the proof for Lemma~\ref{lem:point}, $L[i]$ stores the split subtree in $\Gamma$ where $P_i$ has its probability sum at least $0.5$; $A[i]$ stores the probability sum of $P_i$ in the current split subtree; $B[i]$ stores $w_i\Ed(P_i, Y(T^+_k))$ if $P_i$ belongs to $\calP_k$. We initialize them as the above in Lemma~\ref{lem:point}. 

Next, we traverse each $T^+_k$ of $\Gamma$. For each location $p_{ij}$ on $T^+_k/Y(T^+_k)$, we add $f{ij}$ to $A[i]$. If $A[i]\geq 0.5$ then $P_i$ belongs to $\calP_k$, and thereby we set $L[i]=T^+_k$ in $O(1)$ time. (Note that $T^+$ has no connectors.) After traversing $T^+_k$, we revisit it to reset $A[i]$ as zero for each location $p_{ij}$. Clearly, after the $O(|T^+|)$-time preprocessing, the subset to which any given $P_i$ belongs can be known in $O(1)$ time.  

We proceed with traversing $T^+$ to compute $w_i\Ed(P_i, Y(T^+_k))$ for each $P_i\in\calP_k$ and each $1\leq k\leq s$. More specifically, for each location $p_{ij}$ on $T^+$, we check if $L[i]$ is null and if no, assuming $L[i]=T^+_k$, then we compute $w_if_{ij}\cdot d(p_{ij}, Y(T^+_k))$ and add this value to $B[i]$ in $O(1)$ time. Further, we compute $\tau(\calP_k,Y(T^+_k))$ for each $1\leq k\leq s$ by scanning $L[1\leq i\leq n^+]$. For each $1\leq i\leq n^+$, if $L[i] = T^+_k$, we set $\tau(\calP_k,Y(T^+_k))$ as $B[i]$ if $B[i]$ is larger. Clearly, it takes $O(|T^+|)$ time in total to compute all $\tau(\calP_k,Y(T^+_k))$. 

Now we are ready to determine which split subtree contains a peripheral-center edge. Clearly, for each split subtree $T^+_k$, if $\tau(\calP_k,Y(T^+_k))>\lambda$ and $T^+_k$ contains no vertices with a true flag, then $T^+_k$ must contain a peripheral-center edge. Thus, we add each such $T^+_k$ to a list $Q$. If $Q$ has more than two subtrees, then $\lambda$ is infeasible and thus we are done. If the size of $L$ is exactly two, then we pick the subtree that contains no vertex with a true flag, which can be found in $O(|T^+|)$ time. Assume it is $T^+_k$. Join a vertex to $T^+_k$ for $Y(T^+_k)$ if it is not a vertex. Set the flag of that vertex for $Y(T^+_k)$ as true. We then return this split subtree. If $Q$ contains only one split subtree, assuming it is $T^+_k$, then we apply Lemma~\ref{lem:point} to $Y(T^+_k)$ to decide which of $T^+_k$ and $T^+/\Gamma$ contains a peripheral-center edge in $O(|T^+|)$ time. Note that the flag of $Y(T^+_k)$ is set as true in Lemma~\ref{lem:point} if necessary.  Otherwise, $Q$ is empty and so $T^+/\Gamma$ contains a peripheral-center edge. 

Overall, we can determine which split subtree of $\Gamma$ contains a peripheral-center edge in $O(|T^+|)$ time.\qed
\end{proof}

\section{Computing centers $q^*_1$ and $q^*_2$}\label{sec:alg}
This section presents our algorithm that computes centers $q^*_1$ and $q^*_2$ in $O(mn\log mn)$ time. We say that an edge of $T$ containing a center is a \textit{critical} edge. Similar to the decision algorithm, our algorithm recursively computes each critical edge on $T$ with the assistance of the following key lemma. 

\begin{lemma}\label{lem:keylemma}
Given any point $x$ on $T$, we can decide whether $x$ is a center, and if not, which split subtree of $x$ contains a critical edge; further, if one center is at $x$ then centers $q^*_1$ and $q^*_2$ can be computed in $O(mn)$ time.
\end{lemma}
\begin{proof}
Let $\Gamma$ be the set of all split subtrees $T_1, \cdots, T_s $ of $x$ on $T$. Let $\calP_k$ be the set of uncertain points whose probability sum on $T_k$ are larger than $0.5$. We first compute $w_i\Ed(P_i,x)$ for every $1\leq i\leq n$ and store their expected distances in an auxiliary array $D[1\cdots n]$, which can be done in $O(mn)$ time. Create an additional array $L[1\cdots n]$ to store the split subtree of $\Gamma$ that contains $p_i^*$ of each $P_i$. Another array $F[1\cdots n]$ is initialized as zero for maintaining the probability sums. Initially, $L[1\cdots n]$ are all null and $F[1\cdots n]$ are all zero. We then traverse every split subtree $T_k$ of $\Gamma$. For each location $p_{ij}$ in $T_k$, we add $f_{ij}$ to $F[i]$ and if $F[i]>0.5$ then we set $L[i]$ as that split subtree $T_k$. After traversing $T_k$, we traverse $T_k$ again to reset $F[i]$ as zero for each location $p_{ij}$ in it. Clearly, if $L[i]$ is null then $P_i$'s median $p_i^*$ is at $x$. We proceed with computing the two largest expected distances among all $w_i\Ed(P_i,x)$'s. Suppose $w_1\Ed(P_1,x)\geq w_2\Ed(P_2, x)\geq w_3\Ed(P_3, x)\geq\cdots\geq w_n\Ed(P_n, x)$. We have the following cases.

\begin{enumerate}
    \item If the median $p^*_1$ of $P_1$ is at $x$, i.e., $L[1]$ is null, then we must have $\lambda^*=w_1\Ed(P_1, x)$. It is because $\lambda^*$ is at least $\max_{1\leq i\leq n}w_i\Ed(P_i, p^*_i)$. In this situation, we place both centers at $x$. This case happens if the weight $w_1$ of $P_1$ is relatively much larger than weights of all other uncertain points. 
    
    \item If $p^*_1$ is not at $x$ but $p^*_2$ is at $x$, then we can place a center, e.g., $q^*_2$, at $x$ and so center $q^*_1$ is at $p^*_1$. It follows that the objective value must be $\max_{1\leq i\leq n}w_i\Ed(P_i, p^*_i)$ in that if no centers are at $x$, then the objective value is at least $\max_{1\leq i\leq n}w_i\Ed(P_i, p^*_i)$. Hence, we can set $q^*_2=x$ and compute $q^*_1$ as follows. We first set the weight of every uncertain point as zero except for that of $P_1$ and then call the one-center algorithm~\cite{ref:WangCo17} to compute the center on $T$ with respect to $\calP$ in $O(mn)$ time.
    
    \item If there exists an uncertain point, e.g., $P_3$, so that $w_3\Ed(P_3,x) = w_2\Ed(P_2,x)$ and its median is at $x$, then we can place $q^*_1$ at $p^*_1$ and $q^*_2$ at $x$ and $\lambda^* = \max_{1\leq i\leq n}w_i\Ed(P_i, p^*_i)$. It is easy to see that we can decide in $O(n)$ time if such an uncertain point exists by scanning $D[1\cdots n]$ and $L[1\cdots n]$. 
   
    \item If $p^*_1$ and $p^*_2$ are on different split subtrees, i.e., $L[1]\neq L[2]$, then we have $q^*_1\in L[1]$ and $q^*_2\in x\cup L[2]$. We first decide whether $q^*_2$ is at $x$. We scan $D[1\cdots n]$ to find uncertain points that are of their expected distances at $x$ equal to $w_2\Ed(P_2,x)$. For each such uncertain point, we check whether their medians are exterior of the split subtree containing $p^*_2$. If yes then $q^*_2$ is at $x$, and we then compute $q^*_1$ as the above in an additional $O(mn)$ time. Otherwise, such uncertain points do not exist. Thus, $L[1]$ contains a critical edge and so does $L[2]$. 

    \item If $p^*_1$ and $p^*_2$ are on the same split subtree then a center, e.g., $q^*_1$, must be on $L[1]$. We further decide whether $q^*_2$ is at $x$. If there exist more than one uncertain points, e.g., $P_3$ and $P_4$, with $w_3\Ed(P_3,x)=w_4\Ed(P_4,x)=w_2\Ed(P_2,x)$ but $p_2^*$, $p_3^*$ and $p_4^*$ are on different split subtrees, then $q^*_2$ must be at $x$. It thus follows that $q^*_1$ is at $p^*_1$ and $\lambda = \max_{1\leq i\leq n}w_i\Ed(P_i, p^*_i)$. If all uncertain points with their expected distance equal to $w_2\Ed(P_2,x)$ have their medians on the same split subtree rather than $L[1]$, then $q^*_2$ must be in the interior of that split subtree. Otherwise, these uncertain points have their medians on $L[1]$ or $w_1\Ed(P_1, x) \geq w_2\Ed(P_2, x) > w_3\Ed(P_3, x) \geq \cdots \geq w_n\Ed(P_n, x)$. In this situation, both $q^*_1$ and $q^*_2$ are in the interior of $L[1]$, i.e., $L[1]$ contains both critical edges. 
    
    For this case, we first scan $D[1\cdots n]$ to find uncertain points whose expected distances equal $w_2\Ed(P_2,x)$ but with their medians not on $L[1]$. We create an auxiliary variable $i_1$ to maintain the index of such an uncertain point, and initialize it as $2$. For each $3\leq i\leq n$, we first check whether $D[i] = w_2\Ed(P_2, x)$ and $L[i]\neq L[2]$, and if both are yes, we then check if $i_1$ equals $2$. If it is $2$, then the first uncertain point meeting the above condition is found and thereby we set $i_1 = i$; otherwise, such an uncertain point has already been found and so we continue to compare $L[i_1]$ and $L[i]$. If they are not equal, then there exist more than two uncertain points that have their expected distances at $x$ equal to $w_2\Ed(P_2, x)$ but their medians on different split subtrees, i.e., the first condition holds. Thus, we set $q^*_2 = x$ and so $q^*_1$ is at $p^*_1$, which can be computed in $O(mn)$ time. After scanning $D[1\cdots n]$, we check if $i_1\neq 2$. If yes, then there are only two uncertain points $P_2$ and $P_{i_1}$ so that $w_2\Ed(P_2,x)=w_{i_1}\Ed(P_{i_1},x)$ and their medians are on the different split subtrees. Hence, center $q^*_2$ is in the interior of the split subtree containing $P_{i_1}$'s median, and so we return both $L[1]$ and this split subtree. Otherwise, $L[1]$ contains both centers and is returned. 
\end{enumerate}

It is not hard to see that in $O(mn)$ time we can decide whether $x$ is a center, and if not, which split subtree of $x$ contains a critical edge. \qed
\end{proof}

At the beginning, $T$ is known to contain a critical edge and so we let $T_0=T$. We first compute the centroid $c$ of $T_0$ in $O(|T_0|)$ time and apply Lemma~\ref{lem:keylemma} to $c$ in $O(mn)$ time. Either we obtain two adjacent vertices of $c$ so that the two corresponding split subtrees of $c$ each must have a critical edge, or only one adjacent vertex of $c$ is returned so that both critical edges are in that split subtree of $x$ containing this vertex. For the former case, we set $T_1$ as either subtree, which can be done in $O(|T_0|)$ time, and let the flag of $c$ on $T_1$ as true. In the other case, we let $T_1$ be the only split subtree. Clearly, $|T_1|\leq |T_0|/2$. We continue to search in $T_1$ for a critical edge. First, compute the centroid $c$ of $T_1$ in $O(|T_1|)$ time, and then apply Lemma~\ref{lem:keylemma} to $c$ on $T_1$, which always takes $O(mn)$ time. If two split subtrees are obtained, then the one without any true-flag vertex must contain a critical edge. We thus let $T_2$ be this subtree and set $c$'s flag as true on $T_2$. Otherwise, let $T_2$ be the only split subtree returned. Clearly, $|T_2|\leq |T_1|/2$. We recursively find a critical edge on $T_2$. Clearly, the obtained subtree $T_h$ of the $(h-1)$-th recursive step consists of at most $|T_0|/2^{h-1}$ vertices. Since Lemma~\ref{lem:keylemma} always runs in $O(mn)$ time, the time complexity of each recursive step is $O(mn)$ time. 

After at most $\log mn$ steps, an edge of $T$ remains and it must contain a center. Denote by $e^*_1$ this critical edge. We then adapt the above procedure to find the other critical edge $e^*_2$ on $T$. The only difference is that every recursive step, if two split subtrees are returned by Lemma~\ref{lem:keylemma}, then we always consider the one excluding $e^*_1$. Therefore, $e^*_2$ can be obtained after at most $\log mn$ recursive steps, i.e., in $O(mn\log mn)$ time. At this moment, we can compute $q^*_1$ and $q^*_2$ in $O(mn)$ time by the following lemma. 

\begin{lemma}\label{lem:base}
Given two critical edges $e^*_1$ and $e^*_2$ on $T$, we can find centers $q^*_1$ and $q^*_2$ in $O(mn\log n)$ time.  
\end{lemma}
\begin{proof}
First, we compute the shortest path $\pi$ from $e^*_1$ and $e^*_2$ in $O(|T|)$ time. Let $v_1, v_2, \cdots, v_s$ be all vertices on $\pi$ in order and denote their set by $V$. Assume that $e^*_1$ is edge $e(v_1, v_2)$ and $e^*_2$ is edge $e(v_{s-1}, v_s)$. Let $x$ be any point on $\pi$. We claim that $\lambda^*$ belongs to the set of values $w_i\Ed(P_i,x) = w_j\Ed(P_j,x)$ of all pairs $(P_i, P_j)$ for $x\in \{e^*_1, e^*_2\}$. 

Suppose that $q^*_1$ is on $e^*_1$ and $q^*_1$ is the one determining $\lambda^*$. Removing $e^*_1$ from $T$ generates two subtrees $T_1$ and $T_2$ that respectively includes $v_1$ and $v_2$. Clearly, $q^*_1$ covers all uncertain points whose medians are in $T_1$ under $\lambda^*$. As analyzed in the proof for Lemma~\ref{lem:basecase}, $\Ed(P_i, x)$ of each $P_i\in\calP$ linearly changes as $x$ moves on $e^*_1$ (resp., $e^*_2$) from $v_1$ (resp., $v_{s-1}$) to $v_2$ (resp., $v_s$). So, there must be an uncertain point $P_{i_1}$ with its median interior of $T_1$ and another uncertain point $P_{i_2}$ with its median interior of $T_2$ such that $w_{i_1}\Ed(P_{i_1}, q^*_1) = w_{i_2}\Ed(P_{i_2},q^*_1) = \lambda^*$. Otherwise, we can move $q^*_1$ towards $v_1$ or $v_2$ to reduce the objective value. Thus, our claim is true. 

The proof for Lemma~\ref{lem:basecase} also shows that all $\Ed(P_i, x)$'s for $x\in e^*_1$ (resp., $x\in e^*_2$) can be determined in $O(|T|)$ time. Consider $y=\Ed(P_i, x)$ of each $P_i\in\calP$ respectively for $x\in e_1^*$ and $x\in e_2^*$ in the $x,y$-coordinate system. We can see that values $w_i\Ed(P_i,x) = w_j\Ed(P_j,x)$ of all pairs $(P_i, P_j)$ for $x\in \{e^*_1, e^*_2\}$ belong to the set of the $y$-coordinates of all intersections of the $2n$ lines defined by $y=\Ed(P_i, x)$'s. It follows that $\lambda^*$ is the smallest $y$-coordinate of all intersections with feasible $y$-coordinates. To search for $\lambda^*$ among $y$-coordinates of all intersections, we construct in $O(n\log n)$ time a line arrangement $A$ for the $2n$ lines $y=\Ed(P_i, x)$. The $\lambda^*$ is the $y$-coordinate of the lowest vertex of $A$ with feasible $y$-coordinates. The line arrangement search technique~\cite{ref:ChenAn13} can be employed to find that lowest vertex in $A$ by using our decision algorithm as the oracle to decide the feasibility of the $y$-coordinate of a vertex. 

Then $\lambda^*$ is the $y$-coordinate of the lowest vertex of $A$ with feasible $y$-coordinates. The line arrangement search technique~\cite{ref:ChenAn13} can be employed to find that lowest vertex of $A$ by $O(\log n)$ calls on our decision algorithm. Thus, $\lambda^*$ can be computed in $O(mn\log n)$ time, and centers $q^*_1$ and $q^*_2$ can be computed by one additional call of our decision algorithm on $\lambda^*$. With $O(\log n)$ calls, the lowest vertex, i.e., $\lambda^*$, can be found in $O(mn\log n)$ time. Last, we perform an additional call on our decision algorithm with $\lambda = \lambda^*$ to find $q^*_1$ and $q^*_2$ in $O(mn)$ time. 

Therefore, $q^*_1$ and $q^*_2$ can be computed in $O(mn\log n)$ time. \qed
\end{proof}

As mentioned in Section~\ref{sec:pre}, any given general case can be reduced into a vertex-constrained case in $O(|T|+mn\log mn)$ time. We have the following result. 

\begin{theorem}\label{theorem:1}
The two-center problem of $n$ uncertain points on a tree $T$ can be solved in $O(|T|+ mn\log mn)$ time.
\end{theorem}

\end{document}